%% file: main.tex
\newif\ifFullVersion
\FullVersiontrue

\documentclass{svproc}
\usepackage{microtype,xspace}
\usepackage{color}
\usepackage[usenames,dvipsnames,svgnames,table]{xcolor}
\usepackage[T1]{fontenc}
\usepackage{lmodern}

\usepackage{wrapfig}
\usepackage{graphicx}

\usepackage{mathrsfs}
\usepackage{amssymb,latexsym,amsmath}

\usepackage{amsthm}

\usepackage[T1]{fontenc}
\usepackage{bm}
\usepackage{mathtools} 
\DeclarePairedDelimiter{\ceil}{\lceil}{\rceil}

\usepackage{todonotes}

\usepackage[noend]{algpseudocode}   

\newcommand{\ignore}[1]{}

\newtheorem{observation}{Observation}

\newenvironment{myproof}{\begin{proof}}{\end{proof}}

\DeclareMathOperator{\MST}{\mathrm{MST}}

\DeclareMathOperator{\dep}{\mathrm{dep}}
\DeclareMathOperator{\idx}{\mathrm{idx}}

\DeclareMathOperator{\W}{\mathcal{W}}
\DeclareMathOperator{\myS}{\mathcal{S}}
\DeclareMathOperator{\T}{\mathcal{T}}
\DeclareMathOperator{\mP}{\mathcal{P}}

\newcommand{\sstart}{s_{\mathrm{start}}}
\newcommand{\send}{s_{\mathrm{end}}}
\newcommand{\sleft}{s_{\mathrm{left}}}
\newcommand{\sright}{s_{\mathrm{right}}}
\newcommand{\sbefore}{t_{\mathrm{before}}}
\newcommand{\safter}{t_{\mathrm{after}}}
\newcommand{\nil}{\mbox{\sc nil}\xspace}

\renewcommand{\leq}{\leqslant}
\renewcommand{\geq}{\geqslant}

\newcommand{\Reals}{{\mathbb{R}}}

\newcommand{\etal}{\emph{et al.}\xspace}

\DeclareMathOperator{\speed}{speed}

\usepackage{url}

\begin{document}

\title{Approximation Algorithms for Multi-Robot Patrol-Scheduling with Min-Max Latency}

\author{Peyman Afshani\inst{1} \and Mark de Berg\inst{2} \and Kevin Buchin\inst{2} \and Jie Gao\inst{3} \and Maarten L\"{o}ffler\inst{4} \and Amir Nayyeri\inst{5} \and Benjamin Raichel\inst{6} \and Rik Sarkar\inst{7} \and Haotian Wang\inst{8} \and Hao-Tsung Yang\inst{8}}

\institute{Department of Computer Science, Aarhus University, Denmark \email{peyman@cs.au.dk} \and 
Department of Mathematics and Computer Science, TU Eindhoven, the Netherlands \email{\{M.T.d.Berg, k.a.buchin\}@tue.nl} 
\and Department of Computer Science; Rutgers University; New Brunswick, NJ 08901, USA \email{jg1555@rutgers.edu}
\and Department of Information and Computing Sciences, Utrecht University, the Netherlands \email{m.loffler@uu.nl}
\and School of Electrical Engineering and Computer Science, Oregon State University, OR 97330, USA \email{nayyeria@eecs.oregonstate.edu} 
\and Department of Computer Science; University of Texas at Dallas; Richardson, TX 75080, USA \email{benjamin.raichel@utdallas.edu}
\and School of Informatics, University of Edinburgh, Edinburgh, U.K. \email{rsarkar@inf.ed.ac.uk}
\and Department of Computer Science; Stony Brook University; Stony Brook, NY 11720, USA \email{\{haotwang, haotyang\}@cs.stonybrook.edu}}
\authorrunning{Afshani, de Berg, Buchin, Gao, L\"{o}ffler, Nayyeri, Raichel, Sakar, Wang, Yang}
\titlerunning{Approx. Alg. for Multi-Robot Patrol-Scheduling with Min-Max Latency}


\maketitle              

\begin{abstract}
We consider the problem of finding patrol schedules for $k$ robots to visit a given set of $n$ sites in a metric space. Each robot has the same maximum speed and the goal is to minimize the weighted maximum latency of any site, where the latency of a site is defined as the maximum time duration between consecutive visits of that site. The problem is NP-hard, as it has the traveling salesman problem as a special case (when $k=1$ and all sites have the same weight). We present a polynomial-time algorithm with an approximation factor of $O(k^2 \log \frac{w_{\max}}{w_{\min}})$ to the optimal solution, where $w_{\max}$ and $w_{\min}$ are the maximum and minimum weight of the sites respectively. Further, we consider the special case where the sites are in 1D. When all sites have the same weight, we present a polynomial-time algorithm to solve the problem exactly. If the sites may have different weights, we present a $12$-approximate solution, which runs in  time $(n w_{\max}/w_{\min})^{O(k)}$.
\end{abstract}


%
\input{intro-WAFR}
\section{Problem Definition}

        

As stated in the introduction, our goal is to design a schedule for a set of $k$ robots visiting a set of $n$ sites in such a way that the maximum weighted latency at any of the sites is minimized. It is most intuitive to consider the sites as points in Euclidean space, and the robots as points moving in that space. However, our solutions will actually work in a more general metric space, as defined next. Let $(P,d)$ be a metric space on a set $P$ of $n$ sites, where the distance between two sites~$s_i,s_j \in P$ is denoted by $d(s_i,s_j)$. Consider the undirected complete graph $G=(P,P\times P)$. We view each edge $(s_i,s_j)\in P\times P$ as an interval of length~$d(s_i,s_j)$---so each edge becomes a continuous 1-dimensional space in which the robot can travel---and we define $C(P,d)$ as the continuous metric space obtained in this manner.  From now on, and with a slight abuse of terminology, when we talk about the metric space $(P,d)$ we refer to the continuous metric space $C(P,d)$.
\medskip

Let $R:=\{r_1,\dots, r_k\}$ be a collection of robots
moving in a continuous metric space~$C(P,d)$. We assume without loss of generality
that the maximum speed of the robots is~1. A \emph{schedule} for a robot~$r_j$ is a
continuous function $f_j:\Reals^{\geq 0}\rightarrow C(P,d)$, where $f_j(t)$ specifies 
the position of~$r_j$ at time~$t$. A schedule must obey the speed constraint, that is,
we require $d(f_j(t_1),f_j(t_2))\leq |t_1-t_2|$ for all~$t_1,t_2$. A \emph{schedule 
for the collection~$R$ of robots}, denoted~$\sigma(R)$, is a collection of schedules~$f_j$, 
one for each robot in~$r_j\in R$. (We allow robots to be at the same location at the same time.)
We call the schedule of a robot~$r_j$ \emph{periodic} 
if there exists an offset $t^*_j\geq 0$ and period length $\tau_j>0$
such that for any integer $i\geq 0$ and any $0\leq t<\tau_j$ 
we have $f_j(t^*_j + i\tau_j +t)=f_j(t^*_j + (i+1)\tau_j +t)$. A schedule $\sigma(R)$ is periodic 
if there are $t^*_R\geq 0$ and $\tau_R> 0$ such that for any integer~$i>0$ and any $0\leq t<\tau_R$
we have $f_j(t^*_R + i\tau_R +t)=f_j(t^*_R + (i+1)\tau_R +t)$ for all robots~$r_j\in R$.
It is not hard to see that in the case that all period lengths are rational, $\sigma(R)$ is periodic if and only if the schedules of all robots are periodic. 

We say that a site $s_i\in P$ is \emph{visited} at time~$t$ 
if $f_j(t)=s_i$ for some robot~$r_j$. Given a schedule $\sigma(R)$, the \emph{latency} $L_i$ 
of a site $s_i$ is the maximum time duration during which~$s_i$ is not visited by any robot. 
More formally,
\[
L_i = \sup_{0\leq t_1<t_2} \{ |t_2-t_1| : \mbox{$s_i$ is not visited during the time interval $(t_1,t_2)$} \}
\]
%
%
We only consider schedules where the latency of each site is finite. Clearly such
schedules exists: if $T_{\mathrm{opt}}$ denotes the length of an optimal TSP tour 
for the given set of sites, then we can always get a schedule where $L_i = T_{\mathrm{opt}}/k$
by letting the robots traverse the tour at unit speed at equal distance from each other. Given a metric space~$(P,d)$ and a collection $R$ of $k$ robots, 
the \emph{(multi-robot) patrol-scheduling problem} is to find a schedule~$\sigma(R)$ minimizing the \emph{weighted latency} 
$L := \max_i w_i L_i$, where site $i$ has weight $w_i$ and maximum latency $L_i$.  

Note that it never helps to move at less than the maximum speed between sites---a robot may just as well move at maximum speed and then wait for some time at the next site. Similarly, it does not help to have a robot start at time $t=0$ ``in the middle'' of an edge. Hence, we assume without loss of generality that each robot starts at a site and that at any time each robot is either moving at maximum speed between two sites or it is waiting at a site. 


\input{General_Metric.tex}
\input{one_dimension.tex}
\input{new-approx.tex}

\section{Conclusion and Future Work}
This is the first paper that presents approximation algorithms for multi-robot patrol scheduling minimizing maximum weighted latency in a metric space. The obvious open problem is to improve the approximation ratios for both the general metric setting and the 1D setting.

\begin{small}
\noindent\textbf{Acknowledgement:} Gao, Wang and Yang would like to acknowledge supports from NSF CNS-1618391, DMS-1737812, OAC-1939459. Raichel would like acknowledge support from NSF CAREER Award 1750780.
\end{small}

\bibliographystyle{abbrv}
\bibliography{ref_patrol} 

\ifFullVersion
    \newpage
    \appendix
    \section*{Appendix}
    \input{unweight_general} 
    \input{proof_general}
    \input{proof_1D}
    \input{weighted_frequecy} 
    \input{proof_12_approx} 
    
\else 
\fi

\end{document}

%% file: intro-WAFR.tex
\section{Introduction}
Monitoring a given set of locations over a long period of time has many applications, ranging from infrastructure inspection and data collection to surveillance for public or private safety. 
Technological advances have opened up the possibility to perform these tasks using autonomous robots. 
To deploy the robots in the most efficient manner is not easy, however, and gives rise to interesting algorithmic challenges. This is especially true when
multiple robots work together in a team to perform the task.

We study the problem of finding a \emph{patrol schedule} for a collection of $k$~robots that together monitor a given set of $n$ sites in a metric space, 
where $k$ is a fixed parameter. Each robot has the same maximum speed---from now on assumed to be \emph{unit speed}---and each site has a weight. The goal is to minimize the 
maximum weighted latency of any site. Here the \emph{latency} of a site is defined as the maximum time duration between consecutive visits of that site (multiplied by its weight). A patrol schedule specifies for each robot its starting position and an infinitely long schedule describes how the robot moves over time from site to site. 
\medskip

\noindent\textbf{Related Work.} If $k=1$ and all sites have the same weight, the problem reduces to the Traveling Salesman Problem (TSP) because then the optimal patrol schedule is to have the robot repeatedly traverse an optimal TSP tour. Since TSP is NP-hard even in Euclidean space~\cite{papadimitriou1977euclidean}, this means our problem is NP-hard for sites in Euclidean space as well. There are efficient approximation algorithms for TSP, namely, a $(3/2)$-approximation 
for metric TSP~\cite{christofides1976worst} and a polynomial-time approximation scheme (PTAS) for Euclidean TSP~\cite{arora1998polynomial,mitchell}, which carry over to the patrolling problem for the case where $k=1$ and all sites are of the same weight.

Alamdari~\etal~\cite{alamdari2014persistent}
considered the problem with one robot (i.e., $k=1$) and sites of possibly different weights. 
It can then be profitable to deviate from a TSP tour by visiting heavy-weight
sites more often than low-weight sites. Alamdari~\etal provided algorithms for general graphs with either $O(\log n)$ or $O(\log \varrho)$ approximation ratio, where $n$ is the number of sites and $\varrho$ is the ratio of the maximum and the minimum weight. 

For $k>1$ and even for sites of uniform weights, the problem is significantly harder than for a single robot, since it requires careful coordination of the schedules of the individual robots. The problem for $k>1$ has been studied in the robotics literature under various names, including continuous sweep coverage, patrolling, persistent surveillance, and persistent monitoring~\cite{Elmaliach:2008:RMF:1402383.1402397,6094844,yang2019patrol,liujointinfocom2017,6106761,6042503}. The dual problem has been studied by Asghar~\etal~\cite{asghar2019multi} and Drucker~\etal~\cite{drucker2016cyclic}, where each site has a latency constraint and the objective is to minimize the number of robots to satisfy the constraint among all sites. They provide a $O(\log \rho)$-approximation algorithm where $\rho$ is the ratio of the maximum and the minimum latency constraints. When the objective is to minimize the latency, despite all the works for practical settings, we are not aware of any papers that provide worst-case analysis. There are, however, several closely related problems that have been studied from a theoretical perspective. 

The general family of \emph{vehicle routing problems} (VRP)~\cite{dantzig1959truck} asks for $k$ tours, for a given~$k$, that start from a given depot~$O$ such that all customers' requirements and operational constraints are satisfied and the global transportation cost is minimized. There are many different formulations of the problem, such as time window constraints in pickup and delivery, variation in travel time and vehicle load, or penalties for low quality services;
see the monographs by Golden~\etal~\cite{golden2008vehicle} or T\'oth and Vigo~\cite{toth2002vehicle} for surveys.

In particular, the \emph{$k$-path cover} problem aims to find a collection of $k$ paths that cover the vertex set of the given graph such that the maximum length of the paths is minimized. It has a $4$-approximation algorithm~\cite{arkin2006approximations}. The \emph{min-max tree cover} problem is to cover all the sites with $k$ trees such that the maximum length of the trees is minimized. Arkin~\etal~\cite{arkin2006approximations} proposed a $4$-approximation algorithm for this problem, which was improved to a $3$-approximation by Kahni and Salavatipour~\cite{khani2014improved} and to a  $(8/3)$-approximation by Xu~\etal~\cite{xu2013approximation}. 
The \emph{$k$-cycle cover} problem asks for $k$ cycles (instead of paths or trees) to cover all sites. For minimizing the maximum cycle length, there is an algorithm with an approximation factor of $16/3$~\cite{xu2013approximation}. For minimizing the sum of all cycle lengths, there is a $2$-approximation for the metric setting and a PTAS in the Euclidean setting~\cite{khachai2015polynomial,khachay2016polynomial}. Note that all problems above ask for tours visiting each site once (or at most once), while our patrolling problem asks for schedules where each site is visited infinitely often. 

When the patrol tours are given (and the robots may have different speeds), the scheduling problem is termed the \textit{Fence Patrolling Problem} introduced by Czyzowicz et al.~\cite{czyzowicz2011boundary}. 
Given a closed or open fence (a rectifiable Jordan curve) of length $\ell$ and $k$ robots of maximum speed $v_1, v_2, \ldots, v_k >0$ respectively, the goal is to find a patrolling schedule that minimizes the maximum latency $L$ of any point on the fence. Notice that our problem focuses on a discrete set of $n$ sites while the fence patrolling problem focuses on visiting all points on a continuous curve.
For an open fence (a line segment), a simple partition strategy is proposed, in which each robot moves back and forth in a segment whose length is proportional to its speed. The best solution using this strategy gives the optimal latency if all robots have the same speed and a $2$-approximation of the optimal latency when robots have different maximum speeds. Later, the approximation ratio was improved to $\frac{48}{25}$ by Dumitrescu et al.~\cite{dumitrescu2014fence} allowing the robots to stop. Finally, this ratio is improved to $\frac{3}{2}$ by Kawamura and Soejima~\cite{kawamura2015simple} and the speeds of robots are varied in the patrolling process. 

\begin{figure}
\vspace*{-8mm}
	\begin{tabular}{ccc}
		\begin{minipage}{.37\columnwidth}
    		\centering
    		\includegraphics[width=.5\linewidth]{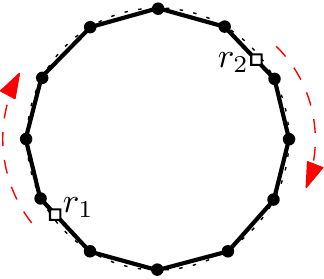}
	\end{minipage}  & 
		\begin{minipage}{.32\columnwidth}
    		\centering
    		\includegraphics[width=.8\linewidth]{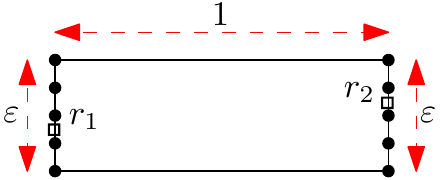}
	\end{minipage}  & 
			\begin{minipage}{.2\columnwidth}
    		\centering
    		\includegraphics[scale = 0.35]{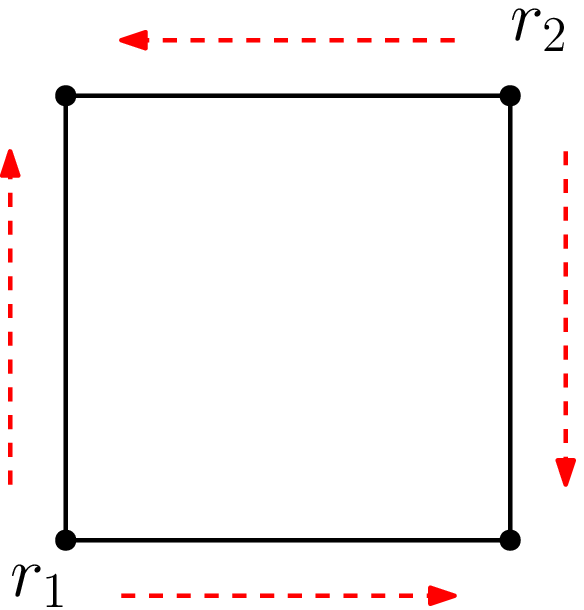}
	\end{minipage} 
	\end{tabular}
\caption{Left: Two robots with $n$ sites evenly placed on a unit circle. The optimal solution is to place two robots, maximum apart from each other, along the perimeter of a regular $n$-gon.
Middle: Two robots with two clusters of vertices of distance $1$ apart. The optimal solution is to have two robots each visiting a separate cluster.
Right: A non-periodic optimal solution.}
\label{fig:example}\vspace*{-4mm}
\end{figure}

%
%


\noindent\textbf{Challenges.} 
For scheduling multiple robots, a number of new challenges arise. One is that already for $k=2$ and all sites of weight $1$ the optimal schedules may have very different structures. For example, if the sites form a regular $n$-gon for sufficiently large $n$, as in Figure~\ref{fig:example} (left), 
an optimal solution would place the two robots at opposite points on the $n$-gon and 
let them traverse the $n$-gon at unit speed in the same direction. 
If there are two groups of sites that are far away from each other,
as in Figure~\ref{fig:example} (middle), it is better to assign each robot 
to a group and let it move along a TSP tour of that group. Figure~\ref{fig:example} (middle)
also shows that having more robots will not always
result in a lower maximum latency. Indeed, adding a third robot in Figure~\ref{fig:example} (middle) 
will not improve the result: during any unit time interval, one of the two groups
is served by at most one robot, and then the maximum latency within that group equals the maximum latency that can already be achieved by two robots for the 
whole problem. 
The two strategies just mentioned---one cycle with all robots evenly placed on it,
or a partitioning of the sites into $k$ cycles, one cycle per robot exclusively---have 
been widely adopted in many practical settings~\cite{elmaliach2009multi,portugal2014finding}. 
Chevaleyre~\cite{chevaleyre2004theoretical} studied the performance of the two strategies but did not provide any bounds. 

Note that the optimal solutions are not limited to the two strategies mentioned above. For example, for three robots it might be best to partition the sites into two groups and assign two robots to one group and one robot to the other group. 
There may even be completely unstructured solutions, that are not even periodic. See Figure~\ref{fig:example} (right) for an example. There are four sites at the vertices of a square with two robots that initially stay on two opposite corners. $r_1$ will choose randomly between the horizontal or vertical direction. Correspondingly, robot $r_2$ always moves in the opposite direction of $r_1$. In this way, all sites have maximum latency $2$ which is optimal. This solution is not described by cycles for the robots, and is not even periodic. Observe that for a single robot, slowing down or temporarily stopping never helps to reduce latency. But for multiple robots, it is not easy to argue that there is an optimal solution in which robots never slow down or stop.

When sites have different weights, intuitively the robots have to visit sites with high weights more frequently than others. Thus, coordination among multiple robots becomes even more complex.
\noindent\textbf{Our results.} We present a number of exact and approximation algorithms which all run in polynomial time. In Section~\ref{General_Metric} we consider the weighted version in the general metric setting and presented an algorithm with approximation factor of $O(k^2\log \frac{w_{\max}}{w_{\min}})$, where $w_{\max}$ and $w_{\min}$ are the maximum weight and minimum weight respectively. The main insight is to obtain a good assignment of the sites to the $k$ robots. We first round up all the weights to powers of two, which only introduces a performance loss by a factor of two. The number of different weights is in the order of $O(\log \frac{w_{\max}}{w_{\min}})$. Given a target maximum weighted latency $L$, we obtain the $t$-min-max tree cover for each set of sites of the same weight $w$, for the smallest possible value $t\leq k$ such that the max tree weight in the tree cover is no greater than $O(L/w)$. 
Then we assign the sites to the $k$ robots sequentially by decreasing weights. Each robot is assigned a depot tree with one of the vertices as the depot vertex. The subset of vertices of a new tree are allocated to existing depots/robots if they are sufficiently nearby; and if otherwise, allocated to a `free' robot. We show that if we fail in any of the operations above (e.g., trees in a $k$-min-max tree cover are too large or we run out of free robots), $L$ is too small. We double $L$ and try again. We prove that the algorithm succeeds as soon as $L \geq L^*$, where $L^*$ is the optimal weighted latency. At that point we can start to design the patrol schedules for the $k$ robots, by using the algorithm in~\cite{alamdari2014persistent}. 

In Section~\ref{sec:1D} we consider the special case where all the sites are points in~$\Reals^1$. When the sites have uniform weights, there is always an optimal solution consisting of $k$ disjoint zigzag schedules (a zigzag schedule is a schedule where a robot travels back and forth along a single fixed interval in~$\Reals^1$), one per robot. Such an optimal solution can be computed in polynomial time by dynamic programming.

When these sites are assigned different weights and the goal is to minimize the maximum weighted latency, we show that there may not be an optimal solution that consists of only disjoint zigzags. Cooperation between robots becomes important. In this case, we turn the problem into the 
Time-Window Patrolling Problem,
the solution to which is a constant approximation to our patrol problem. 
Again we round the weights to powers of two. In the time-window problems, we chop the time axis into time windows of length inversely proportional to the weight of a site -- the higher the weight, the smaller its window size -- and require each site to be visited within its respective time windows. 
This way we have a $12$-approximation solution in time $O((nw_{\max}/w_{\min})^{O(k)})$, where the maximum weight is $w_{\max}$ and the minimum weight is $w_{\min}$.

%% file: General_Metric.tex
\section{Approximation Algorithms in a General Metric}\label{General_Metric}

\def\factor{k^2m}


For sites with weights in a general metric space $(P, d)$, we design an algorithm with approximation factor $O(\factor)$ for minimizing the max weighted latency of all sites by using $k$ robots of maximum speed of $1$, where $m=\log \frac{w_{\max}}{w_{\min}}$. 
Without loss of generality, we assume that the maximum weight among sites is 1. We first round the weight of each site to the least dyadic value and solve the problem with dyadic weights. That is, if node $i$ has weight $w_i$, we take $w'_i=\sup \{2^x | x \in \mathbb{Z} \text{ and } 2^x \geq w_i \}$. Clearly, $w_i\leq w'_i< 2w_i$. This will only introduce another factor of $2$ in the approximation factor on the maximum weighted latency. In the following we just assume the weights are dyadic values. Suppose the smallest weight of all sites is $1/2^{m}$. Denote by $W_j$ the collection of sites of weight $1/2^{j}$. $W_j$ could be empty. Let $\W$ denote the collection of all non-empty sets $W_j$, $0\leq j\leq m$. Note that $|\W| \leq m+1 = \log \frac{w_{\max}}{w_{\min}}+1$. We assume we have a $\beta$-approximation algorithm $\mathcal{A}$ available for the min-max tree cover problem. The currently best-known approximation algorithm has $\beta=8/3$~\cite{xu2013approximation}. 

The intuition of our algorithm is as follows. We first guess an upper bound $L$ on the optimal maximum weighted latency and run our algorithm with parameter $L$. If our algorithm successfully computes a schedule, its maximum weighted latency is no greater than $\beta \factor L$. If our algorithm fails, we double the value of $L$ and run again. We prove that if our algorithm fails, the optimal maximum weighted latency must be at least $L$. Thus, when we successfully find a schedule, its maximum weighted latency is an $O(\factor)$ approximation to the optimal solution. The following two procedures together provide what is needed. 

\begin{itemize}
    \item Algorithm \textsc{$k$-robot assignment}($\W,L$), returns \textsc{False} when there does not exist a schedule with max weighted latency $\leq L$, or, returns $k$ groups: $\T(r_1), \T(r_2), \cdots \T(r_k)$, where $\T(r_i)$ includes a set of trees that are assigned to robot $r_i$. Every site belongs to one of the trees and no site belongs to two trees in the union of the groups. For robot $r_i$, one of the trees in $\T(r_i)$ is called a depot tree $T_{\dep}(r_i) $ and one vertex with the highest weight on the depot tree is a \emph{depot} for $r_i$, denoted by $x_{\dep}(r_i)$. 
    \item With the trees $\T(r_i)$ assigned to one robot $r_i$, Algorithm \textsc{Single Robot Schedule}($\T(r_i)$) returns a single-robot schedule such that every site covered by $\T(r_i)$ has maximum weighted latency $O(\factor \cdot L)$.
\end{itemize}
Denote by $V(T)$ the set of vertices of a tree $T$ and by $d(s_i, s_j)$ the distance between two sites $s_i$ and $s_j$. See the pseudo code of the two algorithms.

\smallskip
\noindent\fbox{\scalebox{0.97}{\begin{minipage}{\textwidth}
			{\sc $k$-robot assignment ($\W,L$)}
\begin{algorithmic}[1]
  \label{alg:robot assignment}
  \For{every set $W_j\in \W$}
    \For{$t \leftarrow$ 1 to $k$} 
        \State Run algorithm $\mathcal{A}$ to obtain a $t$-min-max tree cover $\mathcal{C}^j_t$ on $W_j$.
    \EndFor
    \State $q_j \leftarrow$ smallest integer $t$ s.t. the max weight of trees in $\mathcal{C}^j_t$
    is $< \beta \cdot 2^{j} L$ 
    \State If there is no such $q_j$ then return \textsc{False} \label{algo:k-robot-schedule_first return false}
    \State $\T(W_j) \leftarrow \mathcal{C}^j_{q_j}$
    \EndFor
  \State Set all robots as ``free'' robots, i.e., not assigned a depot tree.
  \For{$j \leftarrow$ 0 to $m$} \Comment{Assign trees to robots} \label{algo:k-robot-schedule_number of groups}
    \For{every tree $T$ in $\T(W_j)$}
    \label{algo:k-robot-schedule_number of trees in each group}
        \State $Q\leftarrow V(T)$
        \For{every non-free robot $r$} 
            \State Let $j'$ be such that $x_{\dep}(r) \in W_{j'}$
            \State $Q' \leftarrow \{v| v\in Q, d(v,x_{\dep}(r)) \leq k 2^{j'}L \}$
            \label{algo:k-robot-schedule_add tree into non-free robot}
            \State Compute $\MST(Q')$ and assign it to robot $r$.
            \State $Q \leftarrow Q \setminus Q'$ 
        \EndFor
        \If $Q\neq \emptyset$ 
             \If{no free robot} 
         \State Return \textsc{False}.
        \Else 
        \State Pick a free robot $r$ and set $T_{\dep}(r) \leftarrow MST(Q)$ 
        \State Pick an arbitrary vertex $x$ in $T_{\dep}(r)$ and set $x_{\dep}(r) \leftarrow x$
        \EndIf
        \EndIf
    \EndFor
  \EndFor
  \State \parbox[t]{110mm}{For each robot $r_i$, let $\T(r_i)$ be the collection of trees assigned to $r_i$, including its depot tree, and return the collections $\T(r_1),\ldots,\T(r_k)$.}
\end{algorithmic}
\end{minipage}}}
\smallskip

The following observation is useful for our analysis later. 
\begin{lemma}\label{lem:far-away-different-depot-trees}
In \textsc{$k$-robot assignment}($\W,L$), the depots $s_i$ and $s_j$, with $w_i\geq w_j$, for different robots have distance more than $k L/w_i$.
\end{lemma}
\begin{proof}
The depot vertices, in the order of their creation, have non-increasing weight. Thus, we could assume without loss of generality that $s_j$ is the depot that is created later than $s_i$. $s_j$ is more than $k L/w_i$ away from the depot $s_i$.
\end{proof}

\begin{lemma}\label{lem:lowerbound-opt}
Let $s_0,\cdots,s_k$ be $k+1$ depot sites, ordered such that $w_0 \geq\cdots \geq w_k$, defined as in Algorithm \textsc{$k$-robot assignment}($\W,L$). The optimal schedule minimizing the maximum weighted latency for $k$ robots to serve $\{s_0,\cdots,s_k\}$ has weighted latency $L^* \geq  2L$.
\end{lemma}


\begin{proof}
Let $\speed(r,t)$ denote the speed of a robot $r$ at time $t$. Let $S$ be a schedule of latency $L^*$. The proof proceeds in $k$ rounds. The goal of the $p$-th round is to change the schedule into a new schedule that has a stationary robot at site $s_{p-1}$. To keep the latency at $L^*$, we will increase the speed of some other robots. We will show the following claim.
\begin{quotation}
\noindent \emph{Claim.} After the $p$-th round we have a schedule of latency $L^*$ such that 
\begin{enumerate}
    \item there is a stationary robot at each of the sites $s_i$ with $i<p$,
    \item at any time $t$ we have $\sum_{r} \speed(r, t) \leq  k$, where the sum is overall $k$ robots.
\end{enumerate}
\end{quotation}
This claim implies that after the $(k-1)$-th round we have a schedule of latency $L^*$ 
with stationary robots at $s_0, s_1, \cdots, s_{k-2}$, and one robot of maximum speed~$k$ 
serving the sites $s_{k-1}$ and~$s_k$. The distance between these sites is at least~$k L/w_{k-1}$,
so the latency $L^*$ of our modified schedule satisfies $L^* \geq 2k L/k=2L$. This is what is needed in the Lemma. 

The proof of the claim is by induction. Suppose the claim holds after the $(p-1)$-th round. Thus we have a stationary robot at each of the sites $s_0,\cdots,s_{p-2}$, and at any time $t$ we have $\sum_{r} \speed(r, t) \leq k$. Note that for $p=1$, the required conditions are indeed satisfied. Now consider the site $s_{p-1}$. 

Define $\ell_0, \ell_1, \cdots$ to be the moments in time where there is at least one robot at $s_{p-1}$ and all robots present at $s_{p-1}$ are leaving. In other words, $\ell_0, \ell_1, \cdots$ are the times at which $s_{p-1}$ is about to become unoccupied. If no such time exists then there is always a robot at $s_{p-1}$, and so we are done. 
Let $a_1, a_2, \cdots$ be the moments in time where a robot arrives at $s_{p-1}$ while no other robot was present at $s_{p-1}$ just before that time, that is, $s_{p-1}$ becomes occupied. Assuming without loss
of generality that $\ell_0 <a_1$, we have 
\[ 
\ell_0\leq a_1\leq \ell_1 \leq \cdots.
\]
Consider an interval $(\ell_i,a_{i+1})$. By definition $a_{i+1}-\ell_i \leq L^*/w_{p-1}$. 
Let $r$ be a robot leaving $s_{p-1}$ at time $\ell_i$ and suppose $r$ is at position $z$ at time $a_{i+1}$. Let $r'$ be a robot arriving at $s_{p-1}$ at time $a_i$. We modify the schedule such that $r$ stays stationary at $s_{p-1}$, while $r'$ travels to $z$ via $s_{p-1}$. We increase the speed of $r'$ by adding the speed of $r$ to it, that is, for any $t \in (\ell_i,a_{i+1})$ we change the  speed of $r'$ at time $t$ to $\speed(r',t) + \speed(r,t)$. Since $r$ is now stationary at $s_{p-1}$, this does not increase the sum of the robot speeds. Moreover, with this new speed, $r'$ will reach $z$ at time $a_{i+1}$. Finally, observe that this modification does not increase the latency.  Indeed, the sites $s_0, \cdots, s_{p-2}$ have a stationary robot by the induction hypothesis, and all sites $s_{p},\cdots ,s_{k}$ are at distance at least $k L/w_{p-1}$ from $s_{p-1}$ so during $(\ell_i,a_{i+1})$  the robots $r$ and $r'$ did not visit any of these sites in the unmodified schedule.  
\qed
\end{proof}

\noindent\fbox{\scalebox{0.97}{\begin{minipage}{\textwidth}
{\sc Single-Robot-Schedule$(\T=\{T_0, T_1, \cdots, T_{h-1}\})$} \Comment{$T_0$ is the depot tree and $w_0$ is the weight of the vertices in $T_0$. $h\leq km$}
\begin{algorithmic}[1]
%
    \State $\delta \leftarrow 2 k L/w_0$. 
    \For{$i \leftarrow$ 0 to $h-1$}
        \State Compute a tour $D_i$ of length at most $2|T_i|$ on the vertices in $T_i$. 
        \State \parbox[t]{90mm}{Partition $D_i$ into a collection $\mP^i=\{P^i_0, P^i_1, \cdots \}$ of at most $\ceil{2|T_i|/\delta}$ paths such that $|P^i_j| \leq \delta$ for all $j$.}
    \State $\idx(i) \leftarrow 0$ \Comment{$P^i_{\idx(i)}$ is the path in $\mP^i$ to be traversed next}
    \EndFor
    \State Put the robot on the first vertex of path $P^0_0$ and set $i \leftarrow 0$
    \While{\textsc{True}}
        \State Let the robot traverse path $P^i_{\idx(i)}$
        \State $i' \leftarrow (i+1) \mod h$    
        \State Let the robot move from the end of $P^i_{\idx(i)}$ to the start of $P^{i'}_{\idx(i')}$
        \State Set $\idx(i) \leftarrow (\idx(i)+1) \mod |\mP_i|$ and set $i \leftarrow i'$
    \EndWhile
\end{algorithmic}
\end{minipage}}}

The proofs for the following two Lemmas can be found in \ifFullVersion
    the appendix.
\else 
    \cite{afshani2020approximation}.
\fi

\begin{lemma}
\label{lem: L leq L_star}
Given $L$, if \textsc{$k$-robot schedule}($\W,L$) returns \textsc{False} then $L^*\geq L$, where $L^*$ is the optimal maximum weighted latency.
\end{lemma}

\begin{lemma}
\label{lem:short-distance}
If \textsc{$k$-robot schedule}($\W,L$) does not return \textsc{False}, each robot is assigned at most $k(m+1)$ trees and a depot site such that 
\begin{itemize}
    \item one of the trees is the depot tree $T_{\dep}$ which includes a depot $x_{\dep}$. $x_{\dep}$ has the highest weight among all sites assigned to this robot;
    \item all other vertices are within distance $k L/\Bar{w}$ from the depot, where $\Bar{w}$ is the weight of $x_{\dep}$;
    \item each tree $T$ has vertices of the same weight $w$ and the sum of tree edge length is at most $\beta L/w$. 
\end{itemize}
\end{lemma}


Now we are ready to present the algorithm for finding the schedule for robot $r_i$ to cover all vertices in the family of trees $\T(r_i)$, as the output of \textsc{$k$-robot schedule}($\W,L$). We apply the algorithm in~\cite{lingasbamboo,alamdari2014persistent} for the patrol problem with one robot, with the only one difference of handling the sites of small weights. The details are presented in the pseudo code \textsc{Single Robot schedule}($\T$) which takes a set $\T$ of $h$ trees. By Lemma~\ref{lem:short-distance}, there are at most $km$ trees assigned to one robot, i.e., $h \leq km$. For a tree $T$ (a path $P$) we use $|T|$ (resp. $|P|$) as the sum of the length of edges in $T$ (resp. $P$). 

\begin{lemma}\label{lem:latency}
The \textsc{Single Robot Schedule}($\T=\{T_0, T_2, \cdots, T_{h-1}\}$), $h\leq k(m+1)$, returns a schedule for one robot that covers all sites included in $\T$ such that the maximum weighted latency of the schedule is at most $O(\factor \cdot L)$.
\end{lemma}

To analyze the running time, we use the best known $t$-min-max tree cover algorithm~\cite{xu2013approximation} with running time $O(n^2 t^2 \log n+t^5 \log n)$. In Algorithm~\textsc{$k$-robot Assignment}, from line 2 to line 8 it takes time in the order of $O(m n^2 \log n) \cdot (1^2+2^2+ \cdots k^2)= O(m n^2 k^3 \log n)$ (suppose $n\gg k$). From line 9 to line 24, we assign some subset of vertices $Q'$ in each tree to occupied robots. The running time is $O(k(m+1) \cdot n\log n)$, where $O(n \log n)$ is the time to compute the minimum spanning tree for $Q'$ (line 16). The total running time is $O(m n^2 \log n)$ for Algorithm~\textsc{$k$-robot Assignment}. Algorithm~\textsc{Single Robot Schedule} takes $O(n)$ time, since a robot is assigned at most $n$ sites. Thus, given a value $L$, it takes $O( n^2 k^3 m\log n)$ to either generate patrol schedules for $k$ robots with approximation factor $O(\factor)$ or confirm that there is no schedule with maximum weighted latency $L$. 

To solve the optimization problem (i.e., finding the minimum $L^*$) if there are fewer than than $k$ sites,  we put one robot per site. Otherwise, we start with parameter $L$ taking the distance between the closest pair of the $n$ sites, and double $L$ whenever the decision problem answers negatively. The number of iterations is bounded by $\log L^*$. Notice that $L^*$ is bounded, e.g., at most $1/k$-th of the traveling salesman tour length.

\begin{theorem}
The approximation algorithm for $k$-robot patrol scheduling for weighted sites in the general metric has running time $O(n^2k^3m\log n \log L^*)$ with a $O(\factor)$-approximation ratio, where $m=\log \frac{w_{\max}}{w_{\min}}$ with $w_{\max}$ and $w_{\min}$ being the maximum and minimum weight of the sites and $L^*$ is the optimal maximum weighted latency.
\end{theorem}

%% file: one_dimension.tex
\section{Sites in $\Reals^1$ }\label{sec:1D}
In this section we consider the case where the sites are points in~$\Reals^1$. 
We start with a simple observation about the case of a single robot. After that we turn our attention to the more interesting case of multiple robots.

We define the schedule of a robot in $\Reals^1$ to be a \emph{zigzag schedule}, or \emph{zigzag} for short, if the robot moves back and forth along an interval at maximum speed (and only turns at the endpoints of the interval). 
\begin{observation}\label{obs:1d-k1}
Let $P$ be a collection of $n$ sites in $\Reals^1$ with arbitrary weights. Then the zigzag schedule where a robot travels back and forth between the leftmost and the rightmost site in $P$ is optimal for a single robot.
\end{observation}

Next, for multiple
robots, as long as the sites have uniform weights,
we show there is an optimal schedule consisting of disjoint
zigzags. Both proofs are in 
\ifFullVersion
    the appendix.
\else 
    \cite{afshani2020approximation}.
\fi

\begin{theorem}
\label{thm:1D_optimal}
Let $P$ be a set of $n$ sites in $\Reals^1$, with uniform weights, and let $k$ be 
the number of available robots, where $1\leq k\leq n$. Then there exists an optimal 
schedule such that each robot follows a zigzag schedule and the intervals 
covered by these zigzag schedules are disjoint. 
\end{theorem}

With Theorem~\ref{thm:1D_optimal}, the min-max latency problem reduces to the following: Given a set $S$ of $n$ numbers and a parameter $k$, compute the smallest $L$ such that $S$ can be covered by $k$ intervals of length at most $L$. When $S$ is stored in sorted order in an array, $L$ can be computed in $O(k^2 \log^2 n)$ time~\cite[Theorem 14]{abrahamsen2017range}. If $S$ is not sorted, there is a $\Omega (n \log n)$ lower bound in the algebraic computation tree model~\cite{ben1983lower}, since for $k=n-1$ element uniqueness reduces to this problem.

%% file: new-approx.tex

We now turn our attention to sites in $\Reals^1$ with arbitrary weights.
In this setting there may not exist an optimal solution that is composed of 
disjoint zigzags (see
\ifFullVersion
    the appendix  
\else 
    \cite{afshani2020approximation}
\fi 
for details), which makes it difficult to compute an optimal solution. 
Hence, we present an approximation algorithm. Let $\rho := w_{\max}/w_{\min}$
be the ratio of the largest and smallest weight of any of the sites.
Our algorithm has a $12$-approximation ratio and runs in polynomial time when $k$, the number of robots, is a constant, and $\rho$ is polynomial in~$n$. 
More precisely, the running time of the algorithm is $O((\rho n)^{O(k)})$.  

Instead of solving the $k$-robot patrol-scheduling problem directly,
our algorithm will solve a discretized version that is defined as follows.
\begin{itemize}
\item The input is a set $P=\{s_1,\ldots,s_n\}$ of sites in $\Reals^1$, each with a weight~$w_i$
      of the form $(1/2)^{\alpha(i)}$ for some non-negative integer~$\alpha(i)$ and such
      that $1=w_1\geq w_2 \geq \cdots \geq w_n$.
\item Given a value $L>0$, which we call the \emph{window length}, we say that
      a $k$-robot schedule is \emph{valid} if the following holds: each site~$s_i$
      is visited at least once during every time interval of the form $[(j-1) L/w_i, j L/w_i]$, 
      where $j$ is a positive integer. 
      The goal is to find the smallest value~$L$ that admits a valid schedule, and
      to report the corresponding schedule.
\end{itemize}
We call this problem the \emph{Time-Window Patrolling Problem}. 
The following lemma shows that its solution can be used to solve the
patrol-scheduling problem. The proof can be found in \ifFullVersion
    the appendix.  
\else 
    \cite{afshani2020approximation}.
\fi 
\begin{lemma}\label{lem-TWP}
Suppose we have a $\gamma$-approximation algorithm for the $k$-robot Time-Window 
Patrolling Problem that runs in $T(n,k,\rho)$ time. Then there is 
$4\gamma$-approximation algorithm for the $k$-robot patrol 
scheduling problem that runs in $O(n \log n + T(n,k,\rho))$ time.
\end{lemma}

\noindent{\rm\bf An algorithm for the Time-Window Patrolling Problem.}
We now describe an approximation algorithm for the Time-Window Patrolling Problem. 
To this end we define a class of so-called \emph{standard schedules}, and we show
that the best standard schedule is a good approximation to the optimal schedule.
Then we present an algorithm to compute the best standard schedule.
\smallskip

Standard schedules, for a given window length $L$, have length (that is, duration) $L/w_n$ and they are composed of $1/w_n$ so-called atomic $L$-schedules.
An \emph{atomic $L$-schedule} is a schedule~$\theta$ 
that specifies the motion of a single robot during a time interval of length~$L$. 
It is specified by a 6-tuple 
\[
(\sstart(\theta),\send(\theta),\sleft(\theta),\sright(\theta),\sbefore(\theta),\safter(\theta)),
\]
where $\sstart(\theta),\send(\theta),\sleft(\theta),\sright(\theta) \in P\cup\{\nil\}$
and $\sbefore(\theta),\safter(\theta)\in\{0,2L/3,L\}$.
Roughly speaking, $\sstart(\theta),\send(\theta),\sleft(\theta),\sright(\theta)$
denote the first, last, leftmost and rightmost site visited during the time interval,
and $\sbefore(\theta),\safter(\theta)$ indicate how long 
the robot can spend traveling 
before arriving
at $\sstart(\theta)$ resp.~ after leaving~$\send(\theta)$. Next
we define this more precisely. 

There are two types of atomic $L$-schedules. For concreteness we explain the
different types of atomic $L$-schedules for the time interval $[0,L]$,
but remember that an atomic $L$-schedule can be executed during any time
interval of length~$L$. 
\begin{description}
\item[Type~I:] $\sstart(\theta),\send(\theta),\sleft(\theta),\sright(\theta)\in P$,
      and $\sbefore(\theta)=0$ and $\safter(\theta)=2L/3$, where
      $\sleft(\theta)$ and $\sright(\theta)$ are the leftmost and rightmost 
      site among the four sites, respectively. (We allow one or more 
      of these four sites to be identical.)
      A Type~I atomic $L$-schedule specifies the following movement of the robot.
      \begin{itemize}
      \item At time $t=0$ the robot is at site $\sstart(\theta)$.
      \item At time $t=L/3$ the robot is at site $\send(\theta)$.
      \item The robot will visit $\sleft(\theta)$ and $\sright(\theta)$
            during the interval~$[0,L/3]$ using the shortest possible path, which must
            have length at most $L/3$. For example, if
            $\sleft(\theta) < \sstart(\theta) < \send(\theta) < \sright(\theta)$,
            then the robot will use the path 
            $\sstart(\theta)\rightarrow \sleft(\theta) \rightarrow \sright(\theta) \rightarrow \send(\theta)$ and we require
            $|\sstart(\theta)-\sleft(\theta)|+ |\sright(\theta)-\sleft(\theta)| + |\sright(\theta) - \send(\theta)| \leq L/3$.
      \item The robot does not visit any sites during $(L/3,L]$
            but is traveling, towards some site to be visited later. In fact, the robot may pass other sites when it is traveling during $(L/3,L]$ but these events are ignored---they are not counted as visits.
      \end{itemize}
\item[Type~II:] 
      $\sstart(\theta)=\send(\theta)=\sleft(\theta)=\sright(\theta)=\nil$,
      and $\sbefore(\theta)=0$ and $\safter(\theta)=L$. A Type~II atomic $L$-schedule
      specifies the following movement of the robot.
      \begin{itemize}
      \item The robot does not visit any sites during $[0,L]$ but is traveling. Again, the robot may pass over sites during its movement. One way to interpret Type~II atomic schedules is that the robot visits a dummy site at time $t=0$ and then spends the entire interval $(0,L]$ traveling towards some site to be visited in a later time interval. 
      \end{itemize}
\end{description}
Note that $\sbefore(\theta)=0$ in both cases. This will no longer be the case, however, when we start concatenating atomic schedules, as explained next.


Consider the concatenation of $2^h$ atomic $L$-schedules, for some $h\geq 0$, and suppose we execute this concatenated schedule during a time interval~$I$ of the form~$[(j-1)2^hL,j2^hL]$. How the robots travel exactly
during interval~$I$ is important for sites of weight more than~$1/2^h$, since such sites need to be visited multiple times. But sites of weight at most~$1/2^h$ need to be visited at most once during~$I$, and so for those sites it is sufficient to know the leftmost and rightmost visited site. Thus our algorithm will concatenate atomic $L$-schedules in a bottom-up manner. This will be done in rounds, where the $h$-th round will ensure that sites of weight~$1/2^h$ are visited. The concatenated schedule will be represented in a similar way as atomic schedules. Next we describe this in detail.

Let $\myS(h)$ denote the collection of all feasible concatenations of~$2^h$ atomic $L$-schedules. Thus $\myS(0)$ is simply the collection of all atomic schedules, and $\myS(h)$ can be obtained from $\myS(h-1)$ by combining pairs of schedules. A schedule $\theta\in\myS(h)$ will be represented by a 6-tuple
\[
(\sstart(\theta),\send(\theta),\sleft(\theta),\sright(\theta),\sbefore(\theta),\safter(\theta)).
\]
As before, $\sstart(\theta),\send(\theta),\sleft(\theta),\sright(\theta)$ denote the first, last, leftmost and rightmost site visited during the time interval. Furthermore, $\sbefore(\theta)$ indicates how much time 
the robot can spend traveling from another site before  arriving at $\sstart(\theta)$, and $\safter(\theta)$ indicates how much time the robot can spend traveling towards another site after leaving at $\send(\theta)$. The values $\sbefore(\theta)$ and $\safter(\theta)$ can now take larger values than in an atomic $L$-schedule. In particular, 
\[
\sbefore(\theta),\safter(\theta) \in \{((2/3)+i)L : 0\leq i <2^h\}\cup \{iL : 0\leq i \leq 2^h\},
\]
where $\sbefore(\theta)+\safter(\theta)\leq 2^h L$.
Note that certain values may only arise in certain situations. For example,
we can only have $\safter(\theta)=2^h L$ for a schedule that is the concatenation of~$2^h$
atomic $L$-schedules of type~II, which means that
$\sstart(\theta),\send(\theta),\sleft(\theta),\sright(\theta) =\nil$ and
$\sbefore(\theta)=0$. 

We denote the concatenation of two schedules $\theta,\theta'\in\myS(h)$
by $\theta \oplus \theta'$. The  representation of
$\theta \oplus \theta'$ can be computed from the representations of $\theta$ 
and $\theta'$:
\[
\sstart(\theta\oplus \theta') = 
     \left\{ \begin{array}{ll}
                \sstart(\theta') & \mbox{ if $\sstart(\theta) =\nil$}  \\
                \sstart(\theta) &  \mbox{ otherwise} \\
              \end{array}
      \right.
\]    
\[
\send(\theta\oplus \theta') = 
     \left\{ \begin{array}{ll}
                \send(\theta) & \mbox{ if $\send(\theta')= \nil$ }  \\
                \send(\theta') & \mbox{ otherwise} \\
              \end{array}
      \right.
\] 
Furthermore, we have
$\sleft(\theta\oplus \theta')  =   \min( \sleft(\theta),\sleft(\theta') )$
and 
$\sright(\theta\oplus \theta') =  \max( \sright(\theta),\sright(\theta') )$.
Finally,
\[
\sbefore(\theta \oplus \theta') = 
     \left\{ \begin{array}{ll}
                \safter(\theta)+ \sbefore(\theta') & \mbox{ if $\sstart(\theta)=\nil$}  \\
                \sbefore(\theta) &  \mbox{ otherwise} \\
              \end{array}
      \right.
\]
\[
\safter(\theta \oplus \theta') = 
     \left\{ \begin{array}{ll}
                \safter(\theta)+ \safter(\theta') & \mbox{ if $\sstart(\theta')=\nil$}  \\
                \safter(\theta) &  \mbox{ otherwise} \\
              \end{array}
      \right.
\]
Note that not any pair of schedules $\theta,\theta'$ can be combined: 
it needs to be possible to travel from $\send(\theta)$ to $\sstart(\theta')$
in the available time. More precisely, assuming $\send(\theta)\neq\nil$
and $\sstart(\theta')\neq\nil$---otherwise a concatenation is always possible---we
need $d(\send(\theta),\sstart(\theta')) \leq \safter(\theta) + \sbefore(\theta')$.

We now define a \emph{standard $k$-robot schedule for window length~$L$} to be a $k$-robot schedule with the following properties.
\begin{enumerate}
\item[(i)] The schedule for each robot belongs to $\myS(\log (1/w_n))$, i.e., each robot starts at a site at time $t=0$, and is the concatenation of $1/w_n$ atomic $L$-schedules.
\item[(ii)] It is a valid $k$-robot schedule for the Time-Window Patrolling Problem, for the time period $[0,L/w_n]$.
\end{enumerate}
A standard schedule $\sigma$ can be turned it into an infinite cyclic schedule,
by executing $\sigma$ and its reverse schedule~$\sigma^{-1}$ in an alternating fashion.
(In $\sigma^{-1}$ each robot simply executes its schedule in $\sigma$ backward.)
Note that~$\sigma^{-1}$ is a valid schedule since $\sigma$ is valid, and so
the schedule alternating between $\sigma$ and $\sigma^{-1}$ is valid.
The following lemma shows that
the resulting schedule is a good approximation of an optimal schedule
for the Time-Window Patrolling Problem (proof in 
\ifFullVersion
    the appendix).
\else 
    \cite{afshani2020approximation}).
\fi
\begin{lemma}\label{lem:standard-schedule}
Let $L^*$ be the minimum window length that admits a valid schedule for the
Time-Window Patrolling Problem, and let $L$ be the minimum window length
that admits a valid standard schedule. Then $L \leq 3 L^*$.
\end{lemma}

We now present an algorithm that, given a window length~$L$, decides if a standard schedule of window length~$L$ exists. Since such a schedule is the concatenation of $1/w_n$ atomic $L$-schedules we basically generate all possible concatenated schedules iteratively from $\myS(0)$ to $\myS(\log (1/w_n))$. Recall that we need to generate a $k$-robot schedule, that is, a
collection of $k$ schedules (one for each robot). We denote by $\myS_k(h)$ the set of all $k$-robots schedules, where each of the schedules is chosen from $\myS(h)$, such that each site of weight at least~$1/2^h$ is visited at least once by one of the robots. If $\sigma=\langle \theta,\ldots,\theta_k\rangle$ and
$\sigma'=\langle \theta'_1,\ldots,\theta'_k\rangle$
are two $k$-robot schedules, then we use 
$\sigma\oplus\sigma'$ to denote the $k$-robot schedule
$\langle \theta\oplus \theta'_1,\ldots,\theta_k\oplus\theta'_k\rangle$.

Note that the concatenation of one pair of single-robot schedules may be the
same as---or, more precisely, have the same representation as---the
concatenation of a different pair of schedules. This may also result
in $k$-robot schedules that are the same.
To avoid generating too many $k$-robot schedules, our algorithm will keep only one schedule of each representation. Our algorithm is now as follows. 
\\[2mm]
\noindent\fbox{\scalebox{0.97}{\begin{minipage}{\textwidth}
{\sc Construct-Schedule$(P,L)$} 
\begin{algorithmic}[1]
  \State \label{li:myS} $\myS(0) \leftarrow$ $\{$ all possible atomic $L$-schedules  $\}$
  \State \label{li:mySk} $\myS_k(0) \leftarrow $
  \parbox[t]{9.8cm}{$\{$ all possible combinations of $k$ schedules from $\myS(0)$ such that all sites of weight~1 are visited by   at least one of the schedules  $\}$ }
  \For{$h\leftarrow 1$ to $m$}  \Comment{Recall that $w_n=1/2^m$}
    \State $\myS_k(h)\leftarrow \emptyset$
    \For{every pair of $k$-robot schedules $\sigma,\sigma' \in \myS_k(h-1)$} \label{li:for}
        \State \label{li:check} \parbox[t]{10cm}{If $\sigma$ and $\sigma'$ can be concatenated and the resulting 
               schedule visits every site of weight  $1/2^h$ at least once
               then add  $\sigma \oplus \sigma'$ to $\myS_k(h)$.}
    \EndFor
    \State Remove any duplicates from $\myS_k(h)$.
  \EndFor
\State If $\myS_k(m)\neq \emptyset$ then return {\sc yes} otherwise return {\sc no}.
\end{algorithmic}
\end{minipage}}}

The algorithm above only reports if a standard schedule of window length~$L$
exists, but it can easily be modified such that it reports such schedule
if it exists. To this end we just need to keep, for each representation in $\myS(h)$ for the current value of~$h$, an actual schedule. Doing so
will not increase the time-bound of the algorithm. The main theorem in this section is as below, with proof in  
\ifFullVersion
    the appendix).
\else 
    \cite{afshani2020approximation}).
\fi

\begin{theorem}
\label{thm:12-approx}
A 12-approximation of the min-max weighted latency for $n$ sites in $\Reals^1$ with $k$ robots, for a constant $k$, can be found in time $O(n^{8k+1} (w_{\max}/w_{\min})^{4k}\log(n\cdot w_{\max}/w_{\min}))=(n w_{\max}/w_{\min})^{O(k)}$, where
the maximum weight of any site is $w_{\max}$ and the minimum weight is $w_{\min}$.
\end{theorem}

%% file: unweight_general.tex

\section{$O(1)$-approximation for unweighted $k$-robot scheduling}
\label{subsec: appetizer approximation algorithms}
We can obtain an approximation algorithm for the patrol-scheduling problem in general metric 
spaces by making a connection to the \emph{$k$-path cover problem}, which is to find 
$k$ paths covering the $n$ sites such that the maximum length of the paths is minimized.
Suppose we have an $\alpha$-approximation algorithm for the $k$-path cover problem.
Let $c^*$ be the maximum path length in an optimal path cover.
For each of the $k$ paths in the cover, connect the last site 
with the first site to create a tour of length at most~$2\alpha c^*$. Now let the $k$ robots follow 
these $k$ tours, obtaining a schedule~$\sigma(R)$ with maximum latency bounded by~$2\alpha c^*$. 
Note that if $L$ denotes the optimal latency for the patrol-scheduling problem, then $c^*\leq L$. 
Indeed, in a solution of latency~$L$, all sites must be visited during any time 
interval of length $L$, and so the paths followed by the robots
during this interval (which have length at most~$L$) are a valid solution
to the $k$-path cover problem. Thus we obtain
a $2\alpha$-approximation for the patrol scheduling problem.

Similarly, we can solve the patrol-scheduling problem with an extra factor of two in the approximation ratio,
using the \emph{min-max $k$-tree cover problem}, which is to find $k$ disjoint trees to cover the $n$ sites such that the maximum tree weight---the weight of a tree is the sum of its edge weights---is minimized, or the \emph{$k$-min-max cycle cover problem}~\cite{xu2013approximation}, which finds $k$ cycles to cover all sites and the length of the longest cycle is minimized. The proof for both claims is similar to the case of $k$-path cover. 

\begin{lemma}
For $n$ sites in a metric space, an $\alpha$-approximation for the
$k$-min-max tree cover problem gives a $2\alpha$-approximation for the patrol scheduling problem.
\end{lemma}

\begin{proof}
To show the connection, take an optimal patrol schedule $\sigma(R)$ from time $0$ to time~$L$, where $L$ is the latency of $\sigma(R)$. This creates $k$ paths that collectively cover all sites. 
Denote by $\sigma_j$ the visiting sequence of robot $r_j$ within this interval. Starting from $\sigma_1$, we shortcut the paths by removing duplicate visits to the same site. Specifically, the visit by robot $r_i$ to a site $s$ is removed if $s$ has already been 
visited by a robot $r_j$ with $j\leq i$. If a site $s$ is removed with $s'$ and $s''$ to be the preceding and succeeding site respectively, the robot moves directly from $s'$ to $s''$; 
by the triangle inequality, the modified path is not longer. This produces at most $k$ 
disjoint paths that cover all sites, thus a tree cover. 
The weight of each path is at most $L$. Thus $L \geq c^*$, where $c^*$ is the optimal weight of 
a min-max tree cover with $k$~trees. On the other hand, for any $k$-tree cover with 
maximum weight~$c$, we can traverse each tree to create a tour with length no longer than~$2c$. 
Let the $k$ robots follow the $k$ tours, thus obtaining a schedule~$\sigma(R)$ with latency
bounded by~$2c$. Hence, an $\alpha$-approximation for the
$k$-min-max tree cover problem gives a $2\alpha$-approximation for the patrol scheduling problem.
\end{proof}

\begin{lemma}
An $\alpha$-approximation for the
$k$-min-max cycle cover problem gives a $2\alpha$-approximation for the patrol scheduling problem.
\end{lemma}

\begin{proof}
If we take an optimal patrol schedule from time~0 
to~$L$, and ask each robot move back to its starting point, then we get $k$ cycles of
length at most $2L$. Hence, $2L>c^*$, where $c^*$ is the min-max cycle length of 
an optimal $k$-cycle cover. This implies that an $\alpha$-approximation for the
$k$-min-max cycle cover problem gives a $2\alpha$-approximation for the patrol scheduling problem.
\end{proof}
In short, we can obtain algorithms with approximation factor~$2\alpha$, where $\alpha$ is the 
approximation factor for any of the problems, \emph{$k$-path cover}~\cite{arkin2006approximations},
\emph{$k$-min-max tree cover}~\cite{khani2014improved,xu2013approximation}, or 
\emph{$k$-min-max cycle cover}~\cite{xu2013approximation}. 
To the best of our knowledge, the best approximation ratio for any of these problems is~$8/3$ 
(namely for the min-max tree cover problem). In this paper we try to get approximation factors 
for the multi-robot patrol-scheduling problem better than~$16/3$. 

%% file: proof_general.tex
\section{Proofs for Min-Max Weighted Latency in General Metric}
\label{sec:proof_general}

\newtheorem*{lem-L-leg-L_star}{\textbf{Lemma \ref{lem: L leq L_star}}}

\begin{lem-L-leg-L_star} {\em 
Given $L$, if \textsc{$k$-robot schedule}($\W,L$) returns \textsc{False} then $L^*\geq L$, where $L^*$ is the optimal maximum weighted latency.}
\end{lem-L-leg-L_star}

\begin{proof}
There are two cases of the algorithm returning \textsc{False}. We discuss them separately.

In the first case, there is a value $j$ such that the maximum tree weight of a $\beta$-approximation of the $t$-min-max tree cover is larger than $\beta 2^{j-1} L$ for all $1\leq t\leq k$ (Line 7).
It implies that the optimal value $\lambda$ of $k$-min-max tree cover is larger than $2^{j-1} L$ for sites in $W_j$. Since the $k$-robot solution also cover all the sites in $W_j$, $\lambda/2^{j-1}$ is also a lower bound of the optimal latency (see 
\ifFullVersion
    the appendix
\else 
    \cite{afshani2020approximation}
\fi 
for details).  Thus, $L^* \geq \lambda/2^{j-1} > 2^{j-1} L/2^{j-1}  =L$.

In the second case, there is a tree with vertices that are far away from existing depots and there is no free robot anymore. Notice that there are precisely $k$ depots at this moment. Suppose the depots are $s_0,s_1, \cdots s_{k-1}$ and there is another vertex $s_{k}$ which is at distance at least $kL/w_i$ from the depot $s_i$ of weight $w_i$, for $0\leq i\leq k-1$. Apply Lemma~\ref{lem:lowerbound-opt}, the latency of the optimal schedule visiting only these $k$ sites is at least $2L$, so is the optimal latency $L^*$. 
\end{proof}

\newtheorem*{lem-short-distance}{\textbf{Lemma \ref{lem:short-distance}}}

\begin{lem-short-distance} {\em 
If \textsc{$k$-robot schedule}($\W,L$) does not return \textsc{False}, each robot is assigned at most $k(m+1)$ trees and a depot site such that 
\begin{itemize}
    \item one of the trees is the depot tree $T_{\dep}$ which includes a depot $x_{\dep}$. $x_{\dep}$ has the highest weight among all sites assigned to this robot;
    \item all other vertices are within distance $k L/\Bar{w}$ from the depot, where $\Bar{w}$ is the weight of $x_{\dep}$;
    \item each tree $T$ has vertices of the same weight $w$ and the sum of tree edge length is at most $\beta L/w$. 
\end{itemize}}
\end{lem-short-distance}

\begin{proof}
Most of the claims are straight-forward from the algorithm \textsc{$k$-robot schedule}($\W,L$). A tree $T$ assigned to a robot has vertices coming from the vertices of the same tree $T'$ in the min-max tree cover (obtained on Line 4). Thus the vertices have the same weight (say $w$). These vertices are within distance $kL/\Bar{w}$, from the depot $x_{\dep}$, where $\Bar{w}$ is the weight of $x_{\dep}$, by Line 15.
Further, the tree $T$ is always taken as a minimum spanning tree on its vertices. Thus the sum of the edge length on $T$ is no greater than that of the original tree $T'$ (with potentially more vertices), which is no greater than $\beta L/w$, by Line 5.   

It remains to prove that each robot~$r$ is assigned at most $km$ trees. Note that the loop of line~\ref{algo:k-robot-schedule_number of groups} in the algorithm has $m+1$ iterations and each loop of line~\ref{algo:k-robot-schedule_number of trees in each group} has at most $k$ iterations. Moreover, in one iteration of lines~13 to 23 each robot~$r$ is assigned at most one tree: it may be assigned a tree in line~16 when it is already non-free, and in line~22 when it was still free. Hence, $r$ is assigned at most $k(m+1)$ trees.
\end{proof}

\newtheorem*{lem-latency}{\textbf{Lemma \ref{lem:latency}} }

\begin{lem-latency} {\em
The \textsc{Single Robot Schedule}($\T=\{T_0, T_2, \cdots, T_{h-1}\}$), $h\leq k(m+1)$, returns a schedule for one robot that covers all sites included in $\T$ such that the maximum weighted latency of the schedule is at most $O(\factor \cdot L)$.}
\end{lem-latency}

\begin{proof}
By Lemma~\ref{lem:short-distance} the distance between the depot and any other vertices on tree $T_i$ is at most $k L/w_0$, where $w_0$ is the weight of the depot. By triangle inequality, the distance of any two sites (either on the same tree or on different trees) is at most $2k L/w_0= \delta$. Consider any site $s$ and assume $s \in P^i_j$ for some $P^i_j \in \mP^i$.  Let $w_i$ be the weight of the vertices in $T_i$.  Note that some path from $\mP^i$ is visited once every $h$ iterations of the while loop of line 9 to 13, and that the paths from $\mP^i$ are visited in a round-robin fashion. Thus $P^i_j$ (and, hence, site $s$) is visited once every $h \cdot |\mP^i|$ iterations. In one iteration the robot moves over a distance at most $\delta$ in line~10, and over a distance at most $\delta$ in line~12. Hence, the total distance traveled by the robot before returning to s is bounded by $h \cdot |\mP^i | \cdot 2\delta$ , and so the total weighted latency is bounded by
$$
w_i \cdot h \cdot |\mP^i | \cdot 2\delta \leq w_i \cdot  h \cdot \ceil{2|T_i|/\delta} \cdot 2\delta
$$
There are two cases. 
If $|T_i|>\delta$, the above term is at most  $w_i \cdot |T_i| \cdot h \leq 2L\cdot h$. 
If $|T_i| \leq \delta$, the above term is at most $w_i \cdot h \cdot 2 \delta \leq 2kL \cdot h$. 
Since $h \leq k(m+1)$, the weighted latency of $s$ is $O(k^2mL)$. 
\end{proof}

%% file: proof_1D.tex
\newtheorem*{obs-1d-k1}{\textbf{Observation \ref{obs:1d-k1}.} }
\begin{obs-1d-k1} {\em
Let $P$ be a collection of $n$ sites in $\Reals^1$ with arbitrary weights. Then the zigzag schedule where a robot travels back and forth between the leftmost and the rightmost site in $P$ is optimal for a single robot.}
\end{obs-1d-k1}

\begin{myproof}
Let $s_1,\ldots,s_n$ be the sites in $P$, ordered from left to right, 
and let $w_i$ denote the weight of~$s_i$.
Then the weighted latency of~$s_i$ in the zigzag schedule 
is $w_i\cdot \max (2\; d(s_i,s_1),2\; d(s_i,s_n))$. Let $s_{i^*}$
be a site whose weighted latency is maximal, and assume without
loss of generality that $d(s_{i^*},s_1) \geq d(s_{i^*},s_n)$. 
Clearly the minimum weighted latency of a robot that only has to visit $s_1$ and $s_{i^*}$
is at most the minimum weighted latency of a robot
that must visit all sites in~$P$. The former is equal to
$w_{i^*}\cdot 2 \; d(s_{i^*},s_1)$ because the robot must go back and forth between $s_1$ and $s_{i^*}$. Since the zigzag on~$P$ has latency $w_{i^*}\cdot 2 \; d(s_{i^*},s_1)$ as well, it must thus be optimal.
\end{myproof}

\newtheorem*{thm-1D_optimal}{\textbf{Theorem  \ref{thm:1D_optimal}}}
\begin{thm-1D_optimal} {\em
Let $P$ be a set of $n$ sites in $\Reals^1$, with uniform weights, and let $k$ be 
the number of available robots, where $1\leq k\leq n$. Then there exists an optimal 
schedule such that each robot follows a zigzag schedule and the intervals 
covered by these zigzag schedules are disjoint. }
\end{thm-1D_optimal}

\begin{myproof}
Let $r_1,\ldots,r_k$ denote the available robots and assume that initially the robots
are ordered from left to right with ties broken arbitrarily. Let
$f_i(t)$ denote the position of robot $r_i$ at time~$t$.
We may assume that this ordering does not change. That is,  $f_1(t) \leq f_2(t) \leq \cdots \leq f_k(t)$ at any time~$t$.
Indeed, when two robots swap, we can switch their roles so that we keep the original order. 


Let $a_i$ and $b_i$ be the leftmost and rightmost site ever visited by $r_i$, respectively, and define $I_i := [a_i, b_i]$.
The order on the robots implies that $a_i \leq a_j$ for $i<j$.
Now consider an optimal schedule with the above properties, where we assume without loss of generality that each robot is assigned a non-empty interval,
which could be a single point. We will modify this schedule
(if necessary) to obtain an optimal schedule consisting of disjoint zigzags.
First we ensure that $a_i<a_j$ for all $i<j$. Suppose that $a_i=a_j$ for (one or more)~$j>i$. 
Note that at any time $t$ such that $f_j(t)=a_i$ for some $j>i$,
we must also have $f_i(t)=a_i$. Hence, the visits of these robots~$r_j$ 
to $a_i$ are not necessary, and we can modify their schedules so that their leftmost visited
sites are the site immediately to the right of~$a_i$. By doing this repeatedly we
obtain a schedule such that $a_i<a_j$ for all~$i<j$.

We now prove the following statement---note that this statement implies the lemma---by induction on $j$:
\begin{quotation}
\noindent There is an optimal schedule such that, for any $1\leq j \leq k$, we have
(i)~the intervals $I_1,\ldots,I_j$ are disjoint from each other and from the intervals $I_{j+1},\ldots,I_k$, and (ii) each of the robot $r_i$ with $1\leq i\leq j$
follows a zigzag on $I_i$. 
\end{quotation}
First consider the case~$j=1$. Note that $a_1$ is the leftmost site in~$P$
and that $r_1$ is the only robot visiting $a_1$. Since $r_1$ also visits~$b_1$,
the latency of $a_1$ is at least $2(b_1-a_1)$, which is achieved
if we make $r_1$ follow a zigzag along~$I_1$. This zigzag guarantees
a latency $2(b_1-a_1)$ for any site in $I_1$, so there is no need for 
another robot to visit those sites. Hence, we can ensure that the
intervals $I_2,\ldots,I_k$ are strictly to the right of $I_1$,
and so the statement is true for~$j=1$.

Now consider the case~$j>1$. Because $a_j<a_i$ for all $i>j$ we know that 
$a_j$ is not visited by any of the robots $r_i$ with $i>j$. By the induction
hypothesis $a_j$ is not visited by any of the robots~$r_i$ with $i<j$ either.
Hence, $r_j$ is the only robot visiting~$a_i$. Following the same reasoning
as in the case $j=1$ we can thus ensure that $r_j$ follows
a zigzag along $I_j$ and that the intervals $I_{j+1},\ldots,I_k$
are disjoint from $I_j$. Together with the induction hypothesis this proves
the statement for~$j$, thus finishing the proof.
\end{myproof}

%% file: weighted_frequecy.tex
\section{An example of Min-Max Weighted Latency in $\Reals^1$}
\label{sec:weighted in 1D}
For any set of sites in $\Reals^1$ with uniform weights, there is an optimal
schedule consisting of disjoint zigzags. This is no longer true 
for arbitrary weights, however, as shown next. Thus a careful coordination
between the robots is needed in this case. 

\begin{figure}[h]
      \centering
      \includegraphics[scale=0.2]{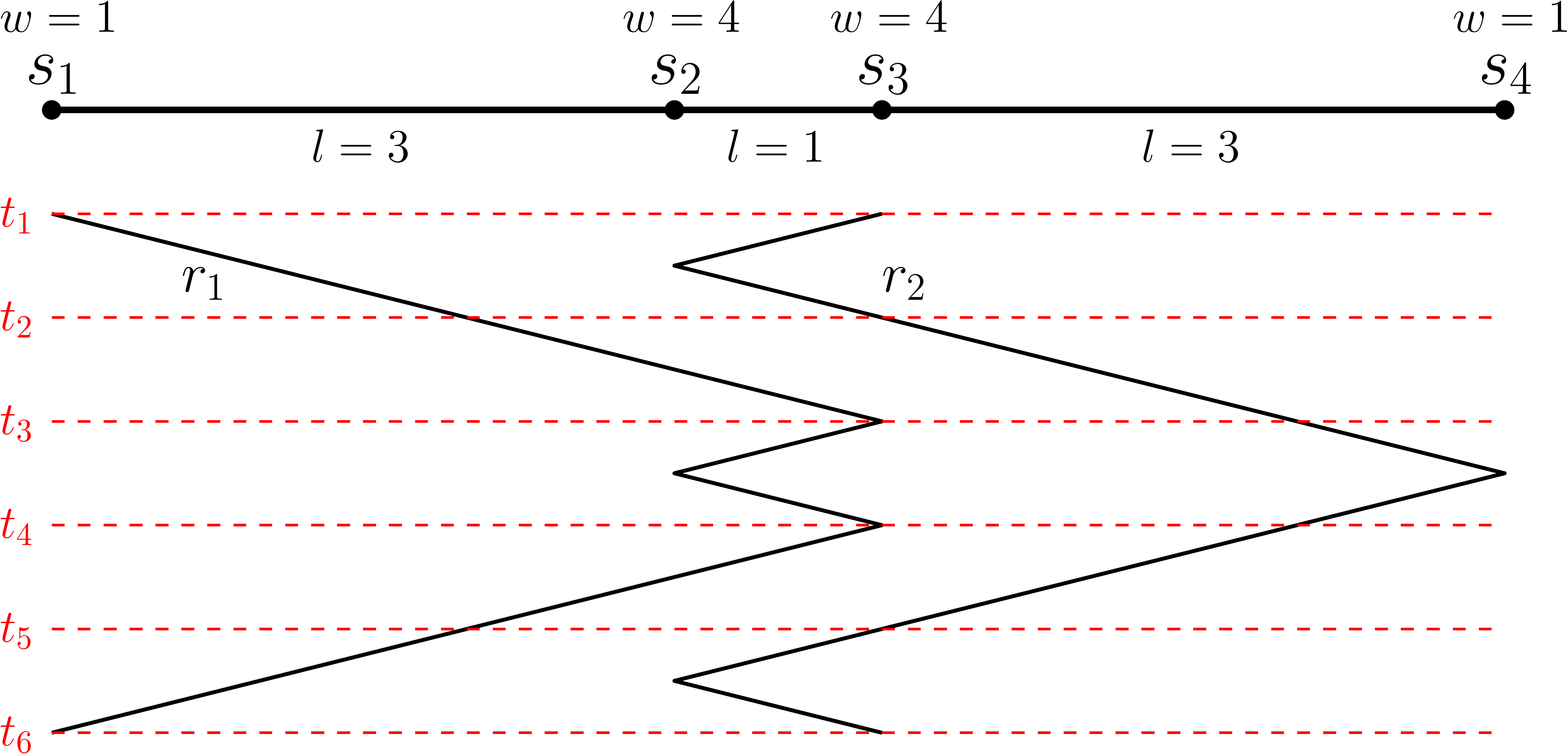}
      \caption{Optimal schedule with maximum weighted latency of $10$.}
      \label{fig:not-independ}
    \end{figure}

Figure~\ref{fig:not-independ} shows an example where a schedule for two robots
consisting of disjoint zigzags is sub-optimal. There are four sites, with $s_1, s_4$ 
having weight~$1$ and $s_2, s_3$ having weight~$4$, and the distances are as shown in the figure. The solution as shown in Figure~\ref{fig:not-independ} has maximum weighted latency of $10$. 

Now, we prove that this solution is optimal. Clearly there is an optimal solution that no robots travel anywhere outside of the interval between $s_1$ and $s_4$ -- if any robot travels to the left (resp. right) of $s_1$ (resp. $s_4$), they just stay at site $s_1$ (resp. $s_4$). We argue that there is an optimal solution in which two robots do not cross each other, if they meet at the same position at the same time and move in opposite directions, we let the two robots turn back at the meeting location. Therefore without loss of generality we assume that at any time robot $r_1$ does not stay to the right of robot $r_2$.  This means that $r_1$ visits $s_1$ and $r_2$ visits $s_4$.

Then, we first prove that $r_1$ needs to visit $s_3$ and $r_2$ needs to visit $s_2$ in the optimal solution. First if $r_2$ stays at $s_4$ at all time, then the max weighted latency at $s_3$ is at least $16$, as $r_1$ needs to travel to $s_1$. Thus $r_2$ must visit $s_3$ as well. Since the distance between $s_3$ and $s_4$ is $3$, it will take $6$ time slots for $r_2$ to visit $s_4$ from $s_3$. Denote this interval as $\Delta=[t, t+6]$.
If there is a schedule with maximum weighted latency less than $10$, $r_1$ should visit $s_3$ at least once in $[t, t+2.5)$ and at least once in $(t+3.5, t+6]$. Between these two visits, $r_1$ should visit $s_2$ at least once, otherwise the maximum weighted latency of $s_2$ exceeds $10$. In this case, $r_1$ must travel at least $10$ units of time back to $s_1$. Thus, our solution of maximum weighted latency of $10$ is optimal.

In a more general case, the distances between $s_1, s_2$ and $s_3, s_4$ are denoted as $x$. Following the schedule in Figure~\ref{fig:not-independ} with $x-1$ times of zigzag between $s_2$ and $s_3$, the schedule has maximum weighted latency of $\max(8, 4x-2)$. There are other cyclic solutions, e.g., one robot performing a zigzag between $s_1, s_3$ and another one doing a zigzag between $s_2, s_4$. The starting location is $s_1$ and $s_2$ respectively. One can verify that the latency of $s_2, s_3$ is $4x$ and the latency of $s_1, s_4$ is $2x+2$. When $x>2$, this cyclic solution performs worse than the solution in the example. In the cyclic solution with disjoint simple zigzags, one robot does a zigzag between $s_1, s_2$ and the other does a zigzag between $s_3, s_4$. The weighted latency is $8x$. With the increment of length between $s_1, s_2$ and $s_3, s_4$, this weighted latency of the best disjoint cyclic solution can become arbitrarily worse.

%% file: proof_12_approx.tex
\section{A $12$-approximation for Min-Max Weighted Latency in $\Reals^1$}
\label{sec:proofs_12_1D}

\newtheorem*{lem-TWP}{\textbf{Lemma \ref{lem-TWP}} }
\begin{lem-TWP} {\em Suppose we have a $\gamma$-approximation algorithm for the $k$-robot Time-Window 
Patrolling Problem that runs in $T(n,k,\rho)$ time. Then there is 
$4\gamma$-approximation algorithm for the $k$-robot patrol 
scheduling problem that runs in $O(n \log n + T(n,k,\rho))$ time.}
\end{lem-TWP}

\begin{myproof}
Consider an instance of the $k$-robot patrol scheduling problem, with sites
$s_1,\ldots,s_n$ and weights~$w_1,\ldots,w_n$. We first sort and scale the weights
such that $1=w_1\geq w_2 \geq \cdots \geq w_n$. Next we replace each weight~$w_i$
by the weight $w'_i$ such that $w'_i$ is of the form $(1/2)^{\alpha(i)}$ for some 
non-negative integer~$\alpha(i)$ and $w_i \leq w'_i \leq 2 w_i$. Then we run the
given $\gamma$-approximation algorithm on the modified input, and report the resulting
schedule. The algorithm obviously runs in the claimed time. It remains to prove
the approximation factor.
\medskip

Let $\sigma$ be an optimal schedule for the $k$-robot patrol 
scheduling problem for the original weights~$w_i$, and let $L^*$ be its weighted latency. 
If we use $\sigma$ with weights~$w'_i$, the weighted latency 
is at most $2L^*$. Let $\sigma'$ be an optimal schedule for
the sites with weights~$w'_i$, and let $L'$ be its weighted latency. 
We have $L' \leq 2L^*$ by the optimality of $\sigma'$.

Now consider $\sigma'$ as a solution for the Time-Window
Patrolling Problem with weights~$w'_i$. Since
the time between any two consecutive visits to a site~$s_i$ in
schedule~$\sigma'$ is at most $L'/w'_i$, the site~$s_i$ must be 
visited during every window of length at least $L'/w'_i$.
Hence, $\sigma'$ is a valid solution for window length~$L'$,
which means that $L'$ is an upper bound on the minimum
window size for the Time-Window 
Patrolling Problem with weights~$w'_i$.

Now suppose we have computed a schedule $\sigma''$ using a
$\gamma$-approximation algorithm for the Time-Window 
Patrolling Problem with weights~$w'_i$, and let~$L''$
be its window length. We have $L'' \leq \gamma L' \leq 2\gamma L^*$. 
Now consider a site $s_i$. Let $L''(s_i)$ denote the weighted latency of $s_i$ 
in~ $\sigma''$. The time between consecutive visits of $s_i$ in $\sigma''$
is at most $2L''/w'_i$, so $L''(s_i) \leq 2 L''$.
The weighted latency of schedule $\sigma''$ can therefore be bounded by
$2L'' \leq 4\gamma L^*$.
\end{myproof}
\newtheorem*{lem-standard-schedule}{\textbf{Lemma \ref{lem:standard-schedule}} }
\begin{lem-standard-schedule} {\em 
Let $L^*$ be the minimum window length that admits a valid schedule for the
Time-Window Patrolling Problem, and let $L$ be the minimum window length
that admits a valid standard schedule. Then $L \leq 3 L^*$.}
\end{lem-standard-schedule}

\begin{myproof}
Let $\sigma$ be a valid $k$-robot schedule of window length~$L^*$ for the Time-Window 
Patrolling Problem. As remarked earlier, we can assume that each robot starts 
at a site. We show how to turn $\sigma$ into a standard schedule of window 
length~$3L^*$, thus proving the lemma.

To turn $\sigma$ into a standard schedule we need to ensure that 
it consists of atomic schedules. Let $\sigma(r_\ell)$ denote the schedule
of robot~$r_\ell$ in~$\sigma$. We modify $\sigma(r_\ell)$
as follows. First we partition $\sigma(r_\ell)$ into $1/w_n$ sub-schedules 
of length~$L^*$. Let $\sigma_j(r_\ell)$ denote the $j$-the sub-schedule,
which is for the time interval~$I_j := [(j-1)L^*,jL^*]$. We modify $\sigma_j(r_\ell)$
into an atomic $3L^*$-schedule $\theta_j(r_\ell)$, as follows.

\begin{itemize}
    \item If $\sigma_j(r_\ell)$ visits at least one sites during~$I_j$, then we modify $\sigma_j(r_\ell)$ into a Type~I atomic schedule
      $(\sstart,\send,\sleft,\sright,0,2L^*)$, where
      $\sstart,\send,\sleft,\sright$ are the first, last, leftmost, and rightmost
      site visited by~$r_\ell$ during~$I_j$. Note that any site visited by
      $\sigma_j(r_\ell)$ is also visited by this atomic schedule. 
      Moreover, $r_\ell$ can indeed travel from $\sstart$ to $\send$
      via $\sleft$ and $\sright$ in time $L^*$, 
      since the distance it has to travel for this is at most the distance traveled
      by $r_\ell$ in $\sigma_j(r_\ell)$. 
\item If $r_\ell$ does not visit any site during the time interval~$t\in[(j-1)L^*,jL^*]$
      then we modify $\sigma_j(r_\ell)$ into a Type~II atomic $3L^*$-schedule
      $(\nil,\nil,\nil,\nil,0,3L^*)$.
\end{itemize}

Since a site~$s_i$ that is visited in~$\sigma_j(r_\ell)$ 
is also visited in~$\theta_j(r_\ell)$, we know that $s_i$ is still
visited in every time interval of the form~$[3(j-1)L^*/w_i,3jL^*/w_i]$, 
as required for a valid schedule of window length~$3L^*$. 

It remains to check that the concatenation of the atomic schedules is feasible. 
That is to say, if $s_i$ is the last site visited by $r_\ell$ during an interval $I_j$
in the schedule~$\sigma$, and $s_{i'}$ is the next site visited by $r_\ell$
in $\sigma$, then we need to show that $r_\ell$ can travel from $s_i$ to $s_{i'}$ 
in the schedule formed by the concatenation of the atomic schedules.
Assume the visit to $s_{i'}$ happens in interval~$I_{j'}$. Then
$d(s_i,s_{i'})\leq (j'-j+1)L^*$, because in $\sigma$ the robot~$r_\ell$ visited $s_i$ 
during $I_j$ and $s_{i'}$ during $I_{j'}$. Note that in $\theta_j(r_\ell)$,
which is of Type~I, we have $2L^*$ time units left after visiting~$s$.
Moreover, the atomic schedules $\theta_{j+1}(r_\ell),\ldots,\theta_{j'-1}(r_\ell)$
are of Type~II and so we have $3L^*$ time units for traveling in each of them.
Hence, we have
\[
2L^* + (j'-j-1)3L^* \geq (j'-j+1)L^* \geq d(s_i,s_{i'})
\]
time units to travel from~$s_i$ to $s_{i'}$, as required.
\end{myproof}

\begin{lemma}\label{lem:12-running-time}
Algorithm~{\sc Construct-Schedule} runs in $O(n^{8k+1} (1/w_n)^{4k})$ time and returns {\sc yes} if and only if the given weighted set $P$ admits a valid standard schedule of window length~$L$.
\end{lemma}

\begin{myproof}
The correctness of the algorithm follows from the discussion above. In particular,
one can show the following by induction on~$h$:
the set $\myS_k(h)$ computed by the algorithm contains all distinct 
representations of $k$-robot standard schedules $\langle \theta_1,\ldots,\theta_k\rangle$
such that
(i) $\theta_j\in \myS(h)$ for all $1\leq j\leq k$, where $\myS(h)$ denotes
the collection of all feasible concatenations of~$2^h$ atomic $L$-schedules, and
(ii) each site $s_i$ of weight $w_i\geq 1/2^h$ is visited at least once by one
of the robots.
Thus $\myS_k(m)$ contains all representations of valid $k$-robot standard schedules.

To prove the time bound, observe that 
\[
|\myS(h)|= O(n^4 (2^h)^2)=O(n^4 2^{2h}),
\]
and so $|\myS_k(h)|= O(n^{4k} 2^{2hk})$.
The check in line~\ref{li:check} of the algorithm takes $O(n)$ time,
assuming $k$ is a constant. The for-loop in lines~\ref{li:for} and~\ref{li:check}
therefore takes 
\[
O\left(\left(n^{4k} 2^{2(h-1)k}\right)^{2} \cdot n \right)
= O\left( n^{8k+1} 2^{4(h-1)k}  \right)
\]
time.
Note that we can also remove duplicates within this time.
Hence, the total time of the algorithm is bounded by
\[
O(kn^{4k+1})  + \sum_{h=1}^{m} O\left( n^{8k+1} 2^{4(h-1)k}  \right)
=  O(n^{8k+1} (2^m)^{4k}) 
=  O(n^{8k+1} (1/w_n)^{4k}),
\]
where the first term is the time for lines~\ref{li:myS} and~\ref{li:mySk}.
\end{myproof}

\newtheorem*{thm-12-approx}{\textbf{Theorem \ref{thm:12-approx}}}
\begin{thm-12-approx} {\em 
A 12-approximation of the min-max weighted latency for $n$ sites in $\Reals^1$ with $k$ robots, for a constant $k$, can be found in time $O(n^{8k+1} (w_{\max}/w_{\min})^{4k}\cdot \log(n\cdot w_{\max}/w_{\min}))=(n w_{\max}/w_{\min})^{O(k)}$, where
the maximum weight of any site is $w_{\max}$ and the minimum weight is $w_{\min}$.}
\end{thm-12-approx}

\begin{myproof}
We use a binary search on a set $\mathcal{L}$ of candidate values of $L$ and find the smallest possible $L$ such that Algorithm~{\sc Construct-Schedule} answers {\sc yes}.
The candidate values in $\mathcal{L}$ are those that cannot be decreased without changing the combinatorial structure of Algorithm~{\sc Construct-Schedule}. Specifically, such critical values are determined in two ways:
\begin{itemize}
    \item The minimum window length that allows for a Type-I atomic schedule with starting/ending/leftmost/rightmost positions at site positions to just fit in $L/3$; there are $O(n^{4})$ such choices. 
    \item $j$ consecutive Type-II atomic schedules that just allow a robot to travel from the last visited site to another site in $\Reals^1$ in time $2L/3+jL$; there are $O(n^2w_{\max}/w_{\min})$ such possibilities, as $j$ can take any integer from $0$ to $w_{\max}/w_{\min}$. 
\end{itemize}
Note that we can also generate these critical values in $O(n^4 + n^2w_{\max}/w_{\min})$ time. 
We can then run a binary search among these set of possible $L$ values for the lowest one for which the decision problem answers positively. The number of iterations in the binary search is bounded by $O(\log(n^4+nw_{\max}/w_{\min}))=O(\log (nw_{\max}/w_{\min}))$. Since the running time of Algorithm~{\sc Construct-Schedule} is $O(n^{8k+1} (w_{\max}/w_{\min})^{4k})$ (Lemma~\ref{lem:12-running-time}), the total running time is 
$$
O(n^{8k+1} (w_{\max}/w_{\min})^{4k}\log(n\cdot w_{\max}/w_{\min}))=(n w_{\max}/w_{\min})^{O(k)}
$$.
\end{myproof}


%% file: main.bbl
\begin{thebibliography}{10}

\bibitem{abrahamsen2017range}
M.~Abrahamsen, M.~de~Berg, K.~Buchin, M.~Mehr, and A.~D. Mehrabi.
\newblock Range-clustering queries.
\newblock In {\em 33rd International Symposium on Computational Geometry (SoCG
  2017), 14-17}, pages 1--16, 2017.

\bibitem{alamdari2014persistent}
S.~Alamdari, E.~Fata, and S.~L. Smith.
\newblock Persistent monitoring in discrete environments: Minimizing the
  maximum weighted latency between observations.
\newblock {\em The International Journal of Robotics Research}, 33(1):138--154,
  2014.

\bibitem{arkin2006approximations}
E.~M. Arkin, R.~Hassin, and A.~Levin.
\newblock Approximations for minimum and min-max vehicle routing problems.
\newblock {\em Journal of Algorithms}, 59(1):1--18, 2006.

\bibitem{arora1998polynomial}
S.~Arora.
\newblock Polynomial time approximation schemes for {Euclidean} traveling
  salesman and other geometric problems.
\newblock {\em Journal of the ACM (JACM)}, 45(5):753--782, 1998.

\bibitem{asghar2019multi}
A.~B. Asghar, S.~L. Smith, and S.~Sundaram.
\newblock Multi-robot routing for persistent monitoring with latency
  constraints.
\newblock In {\em 2019 American Control Conference (ACC)}, pages 2620--2625,
  2019.

\bibitem{ben1983lower}
M.~Ben-Or.
\newblock Lower bounds for algebraic computation trees.
\newblock In {\em Proceedings of the 15th Annual ACM Symposium on Theory of
  Computing}, pages 80--86, 1983.

\bibitem{chevaleyre2004theoretical}
Y.~Chevaleyre.
\newblock Theoretical analysis of the multi-agent patrolling problem.
\newblock In {\em IEEE/WIC/ACM International Conference on Intelligent Agent
  Technology}, pages 302--308, 2004.

\bibitem{christofides1976worst}
N.~Christofides.
\newblock Worst-case analysis of a new heuristic for the travelling salesman
  problem.
\newblock Technical report, Carnegie-Mellon University, 1976.

\bibitem{czyzowicz2011boundary}
J.~Czyzowicz, L.~G{\k{a}}sieniec, A.~Kosowski, and E.~Kranakis.
\newblock Boundary patrolling by mobile agents with distinct maximal speeds.
\newblock In {\em European Symposium on Algorithms}, pages 701--712, 2011.

\bibitem{dantzig1959truck}
G.~B. Dantzig and J.~H. Ramser.
\newblock The truck dispatching problem.
\newblock {\em Management science}, 6(1):80--91, 1959.

\bibitem{drucker2016cyclic}
N.~Drucker, M.~Penn, and O.~Strichman.
\newblock Cyclic routing of unmanned aerial vehicles.
\newblock In {\em International Conference on AI and OR Techniques in
  Constraint Programming for Combinatorial Optimization Problems}, pages
  125--141, 2016.

\bibitem{dumitrescu2014fence}
A.~Dumitrescu, A.~Ghosh, and C.~D. T{\'o}th.
\newblock On fence patrolling by mobile agents.
\newblock {\em The Electronic Journal of Combinatorics}, 21(3):P3--4, 2014.

\bibitem{elmaliach2009multi}
Y.~Elmaliach, N.~Agmon, and G.~A. Kaminka.
\newblock Multi-robot area patrol under frequency constraints.
\newblock {\em Annals of Mathematics and Artificial Intelligence},
  57(3-4):293--320, 2009.

\bibitem{Elmaliach:2008:RMF:1402383.1402397}
Y.~Elmaliach, A.~Shiloni, and G.~A. Kaminka.
\newblock A realistic model of frequency-based multi-robot polyline patrolling.
\newblock In {\em Proceedings of the 7th International Joint Conference on
  Autonomous Agents and Multiagent Systems}, pages 63--70, 2008.

\bibitem{lingasbamboo}
L.~G{\k a}sieniec, R.~Klasing, C.~Levcopoulos, A.~Lingas, J.~Min, and
  T.~Radzik.
\newblock Bamboo garden trimming problem (perpetual maintenance of machines
  with different attendance urgency factors).
\newblock In {\em {SOFSEM} 2017: Theory and Practice of Computer Science},
  pages 229--240, 2017.

\bibitem{golden2008vehicle}
B.~L. Golden, S.~Raghavan, and E.~A. Wasil.
\newblock {\em The vehicle routing problem: latest advances and new
  challenges}, volume~43.
\newblock Springer Science \& Business Media, 2008.

\bibitem{6094844}
L.~Iocchi, L.~Marchetti, and D.~Nardi.
\newblock Multi-robot patrolling with coordinated behaviours in realistic
  environments.
\newblock In {\em 2011 IEEE/RSJ International Conference on Intelligent Robots
  and Systems}, pages 2796--2801, Sept 2011.

\bibitem{kawamura2015simple}
A.~Kawamura and M.~Soejima.
\newblock Simple strategies versus optimal schedules in multi-agent patrolling.
\newblock In {\em International Conference on Algorithms and Complexity}, pages
  261--273, 2015.

\bibitem{khachai2015polynomial}
M.~Y. Khachai and E.~Neznakhina.
\newblock A polynomial-time approximation scheme for the {Euclidean} problem on
  a cycle cover of a graph.
\newblock {\em Proceedings of the Steklov Institute of Mathematics},
  289(1):111--125, 2015.

\bibitem{khachay2016polynomial}
M.~Khachay and K.~Neznakhina.
\newblock Polynomial time approximation scheme for the minimum-weight $k$-size
  cycle cover problem in {Euclidean} space of an arbitrary fixed dimension.
\newblock {\em IFAC-PapersOnLine}, 49(12):6--10, 2016.

\bibitem{khani2014improved}
M.~R. Khani and M.~R. Salavatipour.
\newblock Improved approximation algorithms for the min-max tree cover and
  bounded tree cover problems.
\newblock {\em Algorithmica}, 69(2):443--460, 2014.

\bibitem{liujointinfocom2017}
K.~S. Liu, T.~Mayer, H.~T. Yang, E.~Arkin, J.~Gao, M.~Goswami, M.~P. Johnson,
  N.~Kumar, and S.~Lin.
\newblock Joint sensing duty cycle scheduling for heterogeneous coverage
  guarantee.
\newblock In {\em INFOCOM 2017}, pages 1--9, 2017.

\bibitem{mitchell}
J.~S. Mitchell.
\newblock Guillotine subdivisions approximate polygonal subdivisions: A simple
  polynomial-time approximation scheme for geometric {TSP}, $k$-mst, and
  related problems.
\newblock {\em SIAM Journal on computing}, 28(4):1298--1309, 1999.

\bibitem{papadimitriou1977euclidean}
C.~H. Papadimitriou.
\newblock The {Euclidean} travelling salesman problem is {NP}-complete.
\newblock {\em Theoretical computer science}, 4(3):237--244, 1977.

\bibitem{portugal2014finding}
D.~Portugal, C.~Pippin, R.~P. Rocha, and H.~Christensen.
\newblock Finding optimal routes for multi-robot patrolling in generic graphs.
\newblock In {\em 2014 IEEE/RSJ International Conference on Intelligent Robots
  and Systems}, pages 363--369, 2014.

\bibitem{6106761}
D.~Portugal and R.~P. Rocha.
\newblock On the performance and scalability of multi-robot patrolling
  algorithms.
\newblock In {\em 2011 IEEE International Symposium on Safety, Security, and
  Rescue Robotics}, pages 50--55, Nov 2011.

\bibitem{6042503}
E.~Stump and N.~Michael.
\newblock Multi-robot persistent surveillance planning as a vehicle routing
  problem.
\newblock In {\em Automation Science and Engineering (CASE), 2011 IEEE
  Conference on}, pages 569--575, Aug 2011.

\bibitem{toth2002vehicle}
P.~Toth and D.~Vigo.
\newblock {\em The vehicle routing problem}.
\newblock SIAM, 2002.

\bibitem{xu2013approximation}
W.~Xu, W.~Liang, and X.~Lin.
\newblock Approximation algorithms for min-max cycle cover problems.
\newblock {\em IEEE Transactions on Computers}, 64(3):600--613, 2013.

\bibitem{yang2019patrol}
H.-T. Yang, S.-Y. Tsai, K.~S. Liu, S.~Lin, and J.~Gao.
\newblock Patrol scheduling against adversaries with varying attack durations.
\newblock In {\em Proceedings of the 18th International Conference on
  Autonomous Agents and Multi-Agent Systems}, pages 1179--1188, 2019.

\end{thebibliography}
